\documentclass[letterpaper, 10pt, conference]{ieeeconf}

\IEEEoverridecommandlockouts            

\usepackage{amsmath}
\usepackage{amssymb}
\usepackage{physics}

\usepackage{optidef}

\usepackage{amsthm}

\usepackage{hyperref}

\usepackage{cleveref}

\usepackage{subcaption}

\DeclareMathOperator*{\E}{\mathbb{E}}

\let\Pr\relax
\DeclareMathOperator{\Pr}{\mathrm{Pr}}
\newcommand{\prob}[1]{\Pr(#1)}

\DeclareMathOperator{\quant}{\mathit{Q}}
\newcommand{\quantile}[2]{\quant(#1; #2)}
\newcommand{\proj}[0]{\Pi}

\theoremstyle{definition}
\newtheorem{definition}{Definition}
\newtheorem{assumption}{Assumption}

\theoremstyle{plain}
\newtheorem{proposition}{Proposition}
\newtheorem{theorem}{Theorem}
\newtheorem{lemma}{Lemma}

\theoremstyle{remark}
\newtheorem{remark}{Remark}

\crefname{assumption}{Assumption}{Assumptions}
\crefname{equation}{}{}
\Crefname{equation}{Equation}{Equations}
\crefname{lemma}{Lemma}{Lemma}
\crefname{definition}{Definition}{Definitions}
\crefname{theorem}{Theorem}{Theorems}
\crefname{proposition}{Proposition}{Propositions}

\title{\LARGE \bf
Conformal Prediction-Based MPC for Stochastic Linear Systems
}

\newif\ifarXiv
\arXivtrue      

\author{Lukas Vogel$^{1}$,
        Andrea Carron$^{1}$,
        Eleftherios E. Vlahakis$^{2}$,
        and Dimos V. Dimarogonas$^{3}$%
\thanks{This work was partially supported by
            the Swedish Research Council (VR),
            the Knut \& Alice Wallenberg Foundation (KAW), and
            the Horizon Europe Grant SymAware.%
       }
\thanks{$^{1}$L. Vogel and A. Carron are with the
                Institute for Dynamic Systems and Control,
                ETH Zürich, 8092 Zürich, Switzerland.
                {\tt\small \{vogellu, carrona\}@ethz.ch}.}%
\thanks{$^{2}$E. E. Vlahakis is with the
                Division of Computer Engineering and Computer Science,
                University West, 46132 Trollhättan, Sweden.
                {\tt\small eleftherios.vlahakis@hv.se}.}%
\thanks{$^{3}$D. V. Dimarogonas is with the
                Division of Decision and Control Systems,
                School of Electrical Engineering and Computer Science,
                KTH Royal Institute of Technology, 10044 Stockholm, Sweden.
                {\tt\small dimos@kth.se}.}%
}

\begin{document}

\maketitle
\thispagestyle{empty}
\pagestyle{empty}


\begin{abstract}
We propose a stochastic model predictive control (MPC) framework for linear systems subject to joint-in-time chance constraints under unknown disturbance distributions. Unlike existing approaches that rely on parametric or Gaussian assumptions, or require expensive offline computation, the method uses conformal prediction to construct finite-sample confidence regions for the system's error trajectories with minimal computational effort. These probabilistic sets enable relaxation of the joint-in-time chance constraints into a deterministic closed-loop formulation based on indirect feedback, ensuring recursive feasibility and chance constraint satisfaction. Further, we extend to the output feedback setting and establish analogous guarantees from output measurements alone, given access to noise samples. Numerical examples demonstrate the effectiveness and advantages compared to existing approaches.

\end{abstract}


\section{INTRODUCTION}

Model Predictive Control (MPC) \cite{rawlings2020_model} is a well-established
approach for controlling complex dynamical systems where performance and safety
requirements are paramount. When the system model is only partially known or
subject to unknown external disturbances, a key objective is the satisfaction of
constraints under uncertainty. Robust MPC approaches~\cite{chisci2001_systems}
ensure constraint satisfaction for worst-case disturbances, but may be overly
conservative if large disturbances occur only rarely, and are ill-suited if the
disturbance distribution has unbounded support. Stochastic MPC (SMPC), on the
other hand, incorporates distributional information and relaxes safety
requirements through chance constraints \cite{farina2016_stochastic}, which are
generally non-convex~\cite{paulson2019_efficient, paulson2020_stochastic}.
Existing approaches often employ constraint-tightening relaxations, exploiting
probabilistic reachable sets (PRS)~\cite{hewing2020_recursively} or
mean-variance concentration inequalities \cite{farina2015_approach}, but are
typically confined to Gaussian uncertainty or lead to overly conservative
behavior under joint-in-time chance constraints.

Compared to analytical techniques, sampling-based approaches handle cases where
distributional information is limited or completely unknown, and usually derive
from scenario optimization (SO)~%
\cite{campi2011_samplinganddiscarding, campi2019_scenario},
where probabilistic constraints are relaxed using samples
\cite{hewing2020_scenariobased}.

Scenario-based methods can handle data-driven uncertainty quantification in a
non-conservative manner; however, formal guarantees require the relaxed program
to be convex and often entail significant computational effort.
Conformal prediction (CP)
\cite{vovk2005_algorithmic, angelopoulos2023_conformal} is an alternative
distribution-free framework from machine learning that provides a powerful means
to quantify uncertainty directly from data under minimal assumptions.
It provides tight prediction
regions, satisfying a coverage guarantee from finite samples~%
\cite{angelopoulos2023_conformal}, and has recently been applied to control~%
\cite{lindemann2025formal}. 
 
Here, we use CP to address an SMPC problem under joint-in-time chance
constraints under unknown probability distributions, exploiting its
distribution-independence and low computational requirements.

Closely related to our work are SMPC approaches where SO and CP have been
explored for data-driven handling of chance constraints. Joint-in-time chance
constraints have been considered in~\cite{prandini2012_randomized}, but the
associated online constraint-discarding procedures are computationally
prohibitive for real-time implementation. In \cite{hewing2020_scenariobased}, SO
is used to construct probabilistic sets for the error system, performing
constraint discarding offline, however, assuming stepwise chance constraints and
state measurements are available. Other approaches involve risk allocation~%
\cite{paulson2019_efficient, paulson2020_stochastic} or Monte Carlo simulations
to evaluate the residual probability of joint-in-time constraint violations, at
the cost of online sampling~\cite{wang2021_recursive}. A shrinking-horizon setup
using CP is investigated in~\cite{stamouli2024recursively}, but considers
deterministic systems and unsafe regions induced by other uncontrollable agents.
In \cite{vlahakis2024_conformal}, CP was employed to construct probabilistic
sets for the error system based on state availability, relaxing a stochastic
optimal control problem into a finite-horizon open-loop formulation without
recursive feasibility guarantees.

To address these limitations, this paper proposes a computationally lightweight
SMPC framework for linear systems under joint-in-time chance constraints with
unknown disturbance distributions. The approach leverages CP to construct
finite-sample confidence regions for the system’s error trajectories, enabling
a tractable, data-driven relaxation of the joint probabilistic constraints.
We guarantee recursive feasibility for the relaxed receding-horizon control
problem built on indirect feedback~\cite{hewing2020_recursively},
and extend the framework to the output-feedback setting with unknown measurement
noise, assuming a finite set of noise samples. We prove joint-in-time chance
constraint satisfaction in closed loop based only on available output
measurements, opening new data-driven directions for the challenging output-%
feedback case beyond Gaussian settings.

\ifarXiv
All code is available at
\else
All proofs of our technical results are available in
\cite{vogel2025conformalmpc} and the code is available at
\fi
\url{gitlab.ethz.ch/ics/conformal-prediction-smpc.git}.

\section{PRELIMINARIES}

\subsection{Notation}

For a random variable $x$ following a distribution $\mathcal Q^x$, we write
$x \sim \mathcal Q^x$, its expectation as $\E[x]$ and denote by $x^{(j)}$ a
sample of $x$.
The $\alpha$-quantile of $\mathcal Q^x$ is
$\quantile{\alpha}{x} = \inf \{\,q \in \mathbb R \mid \prob{x \geq q} \geq
\alpha\,\}$. For an empirical distribution supported on samples
$x^{(1)}, \dots, x^{(n)}$, we write the quantile as
$\quantile{\alpha}{\{x^{(1)}, \dots,x^{(n)}\}}$.
For $x \in \mathbb R^n$ and a positive definite matrix
$0 \prec A \in \mathbb R^{n \times n}$, we denote by
$\norm{x}_A = \sqrt{x^\top A x}$ the weighted norm.
The Minkowski sum of two sets $\mathcal A, \mathcal B$ is
$\mathcal A \oplus \mathcal B := \{\,a + b \mid a \in \mathcal A,
b \in \mathcal B\,\}$, and the Pontryagin difference is
$\mathcal A \ominus \mathcal B := \{\,x \mid x + b \in \mathcal A,
\forall b \in \mathcal B\,\}$.

\subsection{Conformal Prediction}

Suppose we have $M + 1$ real-valued, independent and identically distributed
(i.i.d.) random variables, $R^{(0)}, R^{(1)}, \dots, R^{(M)}$, and aim to find
a \emph{prediction region}~$\mathcal C$ for
$R^{(0)}$ using $R^{(1)}, \dots, R^{(M)}$ (sometimes referred to as
\emph{calibration dataset}) that satisfies, for $\vartheta \in (0, 1)$,
\begin{equation*}
  \prob{R^{(0)} \in \mathcal C(R^{(1)}, \dots, R^{(M)})} \geq 1 - \vartheta.
\end{equation*}
Then, such a region is given by the set
\begin{equation*}
  \mathcal C = \left\{\,q \in \mathbb R\,\mid q \leq \quantile{1 -
\vartheta}{\{R^{(1)}, \dots, R^{(M)}, \infty\}}\right\}.
\end{equation*}

This result is based on the following lemma:

\begin{lemma}[{\cite[Lemma~1]{tibshirani2019_conformal}}]
    \label{thm:quantile-lemma}
    Let $R^{(0)}$, $R^{(1)}, \dots, R^{(M)}$ be $M + 1$ i.i.d. random variables,
    and $\vartheta \in (0, 1)$. Then,
    \begin{equation}
      \label{eq:cp-coverage}
      \prob{R^{(0)} \leq \quantile{1 - \vartheta}
        {\{R^{(1)}, \dots, R^{(M)}, \infty\}}}
      \geq 1 - \vartheta.
  \end{equation}
\end{lemma}
No further assumption on the $R^{(k)}$ is needed than i.i.d., which can even be
weakened to exchangeability \cite{angelopoulos2023_conformal}. Note that the
probability in \cref{eq:cp-coverage} is taken over the realization of $R^{(0)}$
and the samples $R^{(1)}, \dots, R^{(M)}$ jointly. Unfortunately,
calibration-conditional guarantees of the form $\prob{R^{(0)} \in \mathcal
C(R^{(1)}, \dots, R^{(M)}) \mid R^{(1)}, \dots, R^{(M)}} \geq 1 - \vartheta$
are not obtainable without further assumptions~%
\cite{lindemann2025formal}, but probably approximately correct (PAC)
guarantees of the form
\begin{equation}
    \label{eq:pac-type-guarantee}
  \prob{\prob{R^{(0)} \leq \tilde q} \geq 1 - \vartheta} \geq 1 - \epsilon,
\end{equation}
can be made through appropriate ``relaxing'' on $\tilde q$
\cite{vovk2012_conditional}.

In the following, we refer to marginal guarantees as in \cref{eq:cp-coverage},
but all results can be transformed to the PAC-type form
\cref{eq:pac-type-guarantee}.

\subsection{Problem Setup}

We consider a stochastic linear time-invariant system
\begin{equation}
  \label{eq:dynamics}
  x(t + 1) = A x(t) + B u(t) + w(t), \quad w(t) \sim \mathcal Q^{w(t)},
\end{equation}
with state $x(t) \in \mathbb R^{n_x}$ and input $u(t) \in \mathbb R^{n_u}$,
where the system matrices $A \in \mathbb R^{n_x \times n_x}$ and $B \in \mathbb
R^{n_x \times n_u}$ are known. System \cref{eq:dynamics} is subject to
a disturbance $w(t) \in \mathbb R^{n_x}$ drawn from an unknown distribution
$\mathcal Q^{w(t)}$, where we assume access to $M$ i.i.d. samples of the
disturbance trajectory,
\begin{equation}
    \label{eq:disturbance-trajectories}
    W^{(k)} = (w^{(k)}(0), \dots, w^{(k)}(\bar N - 1)), \quad k = 1, \dots, M.
\end{equation}
In \cref{eq:disturbance-trajectories}, the entries $w^{(k)}(t)$ may be
correlated across time, but the trajectories are i.i.d.
as $W^{(k)} \sim \mathcal Q^W$.

The system is subject to state and input constraints, $x(t)
\in \mathcal X$, and $u(t) \in \mathcal U$, respectively, where $\mathcal X$ and
$\mathcal U$ are convex compact sets with the origin in their
interior. The problem setup is as follows: We aim to control \cref{eq:dynamics}
over the long but finite horizon $\bar N$, starting from the
known initial condition $x(0) = x_0$, and impose a
maximum probability of constraint violation $\vartheta \in (0, 1)$ over the
entire horizon $\bar N \in \mathbb N$, as
\begin{align}
    \label{eq:chance-constraints-x}
    \prob{x(t) \in \mathcal X, \forall t = 1, \dots, \bar N \mid x(0)} &\geq 1
- \vartheta,\\
    \label{eq:chance-constraints-u}
    \prob{u(t) \in \mathcal U, \forall t = 1, \dots, \bar N \mid x(0)} &\geq 1
- \vartheta.
\end{align}

Equations \crefrange{eq:chance-constraints-x}{eq:chance-constraints-u} are in
the form of \emph{joint-in-time} chance constraints, which arise naturally for
example in safety-critical planning where constraint satisfaction is required
at all times. They are known to be hard to encompass and non-convex~%
\cite{farina2016_stochastic, paulson2020_stochastic}.
For brevity, we will drop the conditioning on $x(0)$ in the following.
The objective is to minimize a sum of stage and terminal costs,
subject to the dynamics \cref{eq:dynamics}, the disturbances, and the
chance constraints
\crefrange{eq:chance-constraints-x}{eq:chance-constraints-u}. This yields the
following stochastic optimal control problem:
\begin{mini!}|s| 
    {U}          
    {            
        \mathbb E_W \left[
            V_\text{f}(x(\bar N)) + \sum_{t = 0}^{\bar N - 1} \ell(x(t), u(t))
        \right]
        \label{opt:sftoc-cost}
    }
    {\label{opt:sftoc}} 
    {}          
    \addConstraint{x(t + 1)}{= A x(t) + B u(t) + w(t)}
    \addConstraint{\Pr(x(t)}{\in \mathcal X, \forall t = 1, \dots, \bar N) \geq
1 - \vartheta}
    \addConstraint{\Pr(u(t)}{\in \mathcal U, \forall t = 1, \dots, \bar N) \geq
1 - \vartheta}
    \addConstraint{x(0)}    {= x_0}
    \addConstraint{w(t)}    {\sim \mathcal Q^{w(t)}},
\end{mini!}
with $U = (u(0), \dots, u(\bar N - 1))$, $W = (w(0), \dots, w(\bar N - 1))$.

Due to linearity, one can decouple \cref{eq:dynamics} into nominal and error
dynamics by introducing a nominal state $z(t)$, following deterministic
dynamics, and choosing the pre-stabilizing feedback control law
$u(t) = Ke(t) + v(t)$, where $e(t) = x(t) - z(t)$, yielding
\begin{align}
    \label{eq:nominal-dynamics}
    z(t + 1) &= A z(t) + B v(t),\\
    \label{eq:error-dynamics}
    e(t + 1) &= (A + BK) e(t) + w(t),
\end{align}
with $z(0) = x(0)$, hence $e(0) = 0$, and $K$ such that $A_K := A + BK$ is
Schur stable.

Our approach in the following may be summarized by:
1)~Probabilistically bound the evolution of \cref{eq:error-dynamics} based on
  samples, yielding a more informative uncertainty characterization than in
  \cite{vlahakis2024_conformal}.
2)~Perform appropriate constraint tightening for~\cref{eq:nominal-dynamics} to
  ensure satisfaction of \crefrange{eq:chance-constraints-x}
  {eq:chance-constraints-u}.
3)~Introduce a receding-horizon relaxation of \cref{opt:sftoc} to compute $v(t)$
  in closed loop.
4)~Extend the method to the output feedback case.

\section{MAIN RESULTS}

\subsection{Conformal bounds for the error system}
\label{sec:conformal-error-bounds}

Motivated by the problem setup of joint-in-time chance constraints, we will use
CP in the following to find convex confidence regions for the error system.
Given that $e(0) = 0$ in \cref{eq:error-dynamics}, the error system evolves
independently of $z(t)$, driven only by the disturbance.

By the decoupling of the error system, i.i.d. disturbance trajectories
$\{ W^{(k)} \}_{k = 1}^M$ as in \cref{eq:disturbance-trajectories} give rise to
i.i.d. error trajectories $\{ E^{(k)} \}_{k = 1}^M$ by propagation through
\cref{eq:error-dynamics}, viz.
\begin{equation}
    \label{eq:error-trajectory-from-disturbances}
    \begin{split}
        E^{(k)} &= (e^{(k)}(1), \dots, e^{(k)}(\bar N))\\
                &= (w^{(k)}(0), \dots, \sum_{i = 0}^{\bar N - 1}
                    A_K^{\bar N - i - 1} w^{(k)}(i)).
    \end{split}
\end{equation}

For any $i \neq j$, $E^{(i)}$ and $E^{(j)}$ are independent since they are
generated from independent disturbance trajectories.
They are identically distributed since they share the same dynamics
\cref{eq:error-dynamics} and $w^{(i)}(t), w^{(j)}(t)\sim\mathcal Q^{w(t)}$.
Thus, we may leverage \Cref{thm:quantile-lemma} to find a confidence region for
\cref{eq:error-dynamics}, as is detailed below:
\begin{definition}
    \label{def:confidence-region}
    A confidence region $\mathcal E$ for the error process $E = (e(1), \dots,
    e(\bar N))$ of probability $1-\vartheta$, $\vartheta{\in}(0, 1),$ satisfies
    \begin{equation*}
        \prob{E \in \mathcal E} = \prob{(e(1), \dots, e(\bar N)) \in \mathcal E}
        \geq 1 - \vartheta.
    \end{equation*}
\end{definition}
\begin{lemma}
    \label{thm:cp-confidence-regions}
    Suppose $s\colon (\mathbb R^{n_x})^{\bar N} \to \mathbb R$ is a convex
    function (\emph{score function}) and that a dataset
    $\mathcal D = \{ E^{(k)} \}_{k = 1}^{M}$
    of i.i.d. trajectories of the error process $E = (e(1), \dots, e(\bar N))$
    is available. Let
    \begin{equation*}
        R^{(k)} = s(E^{(k)}), \quad \forall k = 1, \dots, M,
    \end{equation*}
    and pick $\hat q = \quantile{1 - \vartheta}
    {\{R^{(1)}, \dots, R^{(M)}, \infty\}}$
    for some violation probability $\vartheta \in (0, 1)$.
    Then, the set defined by
    \begin{equation*}
        \mathcal E = \{\,Y \in (\mathbb R^{n_x})^{\bar N} \mid s(Y) \leq
        \hat q\,\}
    \end{equation*}
    is a convex conformal confidence region of level $1 - \vartheta$ for a
    newly-drawn sample $E^{(0)}$ of the error process $E$.
\end{lemma}
\ifarXiv
\begin{proof}
    By assumption, any realization $E^{(0)}$ of the random process $E$ is
    i.i.d. with the trajectories in $\mathcal D$.
    Thus, $R^{(0)} = s(E^{(0)})$ is i.i.d. with $\{ R^{(1)}, \dots, R^{(M)}\}$.
    By~\cref{thm:quantile-lemma},
    \begin{equation*}
        \prob{R^{(0)} \leq \hat q} = \prob{E^{(0)} \in \mathcal E}
        \geq 1 - \vartheta.
    \end{equation*}
    Further, sublevel sets of convex functions are convex.
\end{proof}
\fi

Under mild assumptions~\cite{tibshirani2019_conformal}, it can be shown that the
conformal confidence regions are non-conservative, i.e., they satisfy
$1 - \vartheta \leq \prob{E^{(0)} \in \mathcal E} \leq
 1 - \vartheta + \frac{1}{M + 1}$ \cite{angelopoulos2023_conformal}.
\Cref{thm:cp-confidence-regions} provides a \emph{joint-in-time} coverage
guarantee for the entire random process $E$. In contrast, an equivalent
guarantee may be derived from stepwise confidence regions $\mathcal E_t$ for
$e(t)$ through Boole's inequality, but this will in general result in
conservative bounds~\cite{farina2016_stochastic}.
\begin{remark}
    If $\mathcal E$ is a confidence region of level $1 - \vartheta$ for the
    error process $E$, stepwise \emph{projected confidence regions} are needed
    for stepwise constraints, which can be defined by
    \begin{equation}
        \label{eq:projected-confidence-region}
        \mathcal E_t := \proj_t(\mathcal E),
    \end{equation}
    where $\proj_t\colon (\mathbb R^{n_x})^{\bar N} \to \mathbb R^{n_x}$ is the
    projection onto the $t$-th component of the product space
    $(\mathbb R^{n_x})^{\bar N}$. This implies that
    \begin{equation*}
        E \in \mathcal E \implies e(t) \in \mathcal E_t,
            \forall t = 1, \dots, \bar N,
    \end{equation*}
    and therefore we can conclude
    \begin{equation}
        \label{eq:probability-in-tube}
        \prob{e(t) \in \mathcal E_t, \forall t = 1, \dots, \bar N} \geq
        \prob{E \in \mathcal E}.
    \end{equation}
    Hence, the projected sets, when applied stepwise in receding horizon
    control, will preserve the coverage. If $\mathcal E$ is further the
    Cartesian product of $\mathcal E_t$, \cref{eq:probability-in-tube} reduces
    to an equality.
\end{remark}

We next examine the choice of an appropriate score function
for \cref{thm:cp-confidence-regions}. Multiple score functions for CP have been
proposed, including the (weighted) maximum score~%
\cite{vlahakis2024_conformal, cleaveland2024_conformal},
\begin{equation*}
    s(e(1), \dots, e(\bar N)) :=
        \max_{t = 1, \dots, \bar N} \alpha_t \norm{e(t)}.
\end{equation*}
The weights $\alpha_t \geq 0$ scale the confidence region at each time step and
may be optimized; but this requires solving a linear complementarity program,
growing quickly with the number of available samples
\cite{cleaveland2024_conformal}. Further, to retain the i.i.d. requirements of
\cref{thm:quantile-lemma}, this optimization must be over a separate dataset
different from conformal calibration.

We propose a more lightweight approach:
\begin{definition}
    \label{def:mahalanobis-distance}
    For a random variable $e \sim \mathcal Q^e$ with mean $\mu_e$ and covariance
    $\Sigma_e$, the Mahalanobis distance of a point $x$ to $\mathcal Q^e$ is
    \begin{equation*}
        d(x; \mu_e, \Sigma_e) =
            \sqrt{(x - \mu_e)^\top \Sigma_e^{-1} (x - \mu_e)}.
    \end{equation*}
\end{definition}
The Mahalanobis distance can be used as a score function, taking the
trajectory-wise maximum and using sample estimates to obtain a joint-in-time
coverage guarantee, viz.
\begin{equation}
    \label{eq:mahalanobis-score-computation}
    s(E^{(k)}) = \max_{t \in \{ 1, \dots, \bar N \}}
        d(e^{(k)}(t); \hat \mu_t, \hat \Sigma_t).
\end{equation}
To find $\hat \mu_t, \hat \Sigma_t$ for \cref{eq:mahalanobis-score-computation},
we proceed as follows: Partitioning the dataset of i.i.d. samples $\mathcal D$
into $\mathcal D_\text{fit}$ and $\mathcal D_\text{cal}$, initially use
$\mathcal D_\text{fit}$ to compute the sample mean and covariance as
\begin{align}
    \label{eq:empirical-mean}
    \hat \mu_t &= \frac{1}{\abs{\mathcal D_\text{fit}}}
        \sum_{k \in \mathcal D_\text{fit}} e^{(k)}(t),\\
    \label{eq:empirical-covariance}
    \hat \Sigma_t &= \frac{1}{\abs{\mathcal D_\text{fit}} - 1}
        \sum_{k \in \mathcal D_\text{fit}}
        (e^{(k)}(t) - \hat \mu_t) (e^{(k)}(t) - \hat \mu_t)^\top.
\end{align}
Subsequently, use the sample estimates in
\cref{eq:mahalanobis-score-computation} to compute the scores for each
trajectory $E^{(k)} \in \mathcal D_\text{cal}$. Since the trajectories in
$\mathcal D_\text{cal}$ were ``held out'' during the calculation of
$\hat \mu_t, \hat \Sigma_t$, their scores $s(E^{(k)})$ are i.i.d. with the score
of the realization of the true error trajectory $s(E^{(0)})$; thus,
\cref{thm:cp-confidence-regions} may be used with $\mathcal D = 
\mathcal D_\text{cal}$ to find the region
\begin{equation*}
  \label{eq:mahalanobis-region}
  \mathcal E = \{\,E \mid \|e(t) - \hat \mu_t\|_{\hat \Sigma^{-1}_t} \leq
    \hat q, \forall t = 1, \dots, \bar N\},
\end{equation*}
where $\hat q$ is the conformal quantile from \cref{thm:cp-confidence-regions}.

In contrast to the methods from \cite{cleaveland2024_conformal}, finding the
parameters of this score function is computationally lightweight; nevertheless,
the resulting ellipsoidal sets adapt to the mean and variance at each time step.

\subsection{Indirect Feedback SMPC}
\label{sec:indirect-feedback-smpc}

We will now return to the original problem: Find an approximation to
\cref{opt:sftoc} leveraging feedback to control \cref{eq:dynamics} such that
the joint-in-time chance constraints
\crefrange{eq:chance-constraints-x}{eq:chance-constraints-u} are satisfied in
closed loop. To that end, we employ indirect feedback~%
\cite{hewing2020_recursively}: we state the optimization problem using the
nominal dynamics \cref{eq:nominal-dynamics} and ensure constraint satisfaction
through appropriate tightening while introducing feedback from the true state
indirectly from available samples.
The initial state is chosen at each iteration as
\begin{equation}
    \label{eq:if-update-law}
    z_0(t + 1) = z^*_1(t),
\end{equation}
where $z_1^*(t)$ is the predicted nominal state $z(t + 1)$ of the optimal state
trajectory obtained at
time $t$. This initialization preserves feasibility of the optimization problem,
and the closed-loop control law
\begin{equation}
    \label{eq:if-cpmpc-feedback-law}
    u(t) = Ke(t) + v_0^*(t) = K(x(t) - z(t)) + v_0^*(t)
\end{equation}
ensures that the true state tracks the nominal trajectory. The error system
evolves as \cite{hewing2020_recursively}
\begin{equation}
    \label{eq:error-dynamics-cl}
    \begin{split}
        e(t + 1) &= x(t + 1) - z_0(t + 1) = x(t + 1) - z_1^*(t)\\
                 &= A x(t) + B u(t) - A z_0^*(t) - B v_0^*(t)\\
                 &= A_K e(t) + w(t).
    \end{split}
\end{equation}
Note the symmetry to \cref{eq:error-dynamics}: The error in closed loop evolves
\emph{equal in distribution} to the error trajectories obtained from~%
\cref{eq:error-trajectory-from-disturbances}, which allows
leveraging \cref{thm:cp-confidence-regions}: The receding horizon MPC
may ``inherit'' the guarantees from offline designed probabilistic sets
for the error process. Given a confidence region $\mathcal E$ of probability
$1 - \vartheta$ for the error system, chance constraint satisfaction can thus be
ensured through tightening of the state and input constraints as
\begin{align}
    \label{eq:tightened-state-constraints}
    \mathcal Z_t &= \mathcal X \ominus \mathcal E_t, &
         \forall t &= 1, \dots, \bar N,\\
    \label{eq:tightened-input-constraints}
    \mathcal V_t &= \mathcal U \ominus K \mathcal E_t, &
         \forall t &= 1, \dots, \bar N,
\end{align}
where $\mathcal E_t$ is the projected confidence region as in
\cref{eq:projected-confidence-region}. As $\mathcal X,~\mathcal U$ are convex,
the tightened state and input constraints in
\crefrange{eq:tightened-state-constraints}{eq:tightened-input-constraints}
remain convex. If $\mathcal X, \mathcal U$ are additionally polytopes in
half-space representation, the half-spaces may be tightened by the support
function of the ellipsoid.

\begin{proposition}
    \label{thm:constraint-satisfaction}
    Let $\mathcal E$ be a confidence region of level $1 - \vartheta$
    for the error process $E = (e(1), \dots, e(\bar N))$, found from a
    dataset of i.i.d. disturbance trajectories $\{ W^{(k)} \}_{k = 1}^M$ using~%
    \cref{eq:error-trajectory-from-disturbances}, $s$ be a score
    function as in \cref{thm:cp-confidence-regions} and $\mathcal X$,
    $\mathcal U$ be the state and input constraint sets of \cref{eq:dynamics}.
    If $v(t) \in \mathcal V_t$ is chosen such that $z(t) \in \mathcal Z_t$ with
    $\mathcal Z_t$, $\mathcal V_t \neq \emptyset$ as defined in
    \crefrange{eq:tightened-state-constraints}{eq:tightened-input-constraints},
    then the realizations of the closed-loop state and input trajectories
    $X^{(0)} = (x^{(0)}(1), \dots, x^{(0)}(\bar N))$ and 
    $U^{(0)}(t) = (u^{(0)}(0), \dots, u^{(0)}(\bar N - 1))$ satisfy
    \crefrange{eq:chance-constraints-x}{eq:chance-constraints-u}.
\end{proposition}
\ifarXiv
\begin{proof}
We denote by $E^{(0)}$ the error trajectory realized by the closed loop system,
initialized at $e^{(0)}(0) = 0$. It evolves according to
\cref{eq:error-dynamics-cl}, which matches \cref{eq:error-dynamics} in
distribution. Since the realized disturbance trajectory $W^{(0)}$ is drawn
i.i.d. with the calibration trajectories $\{ W^{(k)} \}_{k = 1}^{M}$, the
distribution of $E^{(0)}$ matches the $\mathcal D = \{ E^{(k)} \}_{k = 1}^M$
following \cref{eq:error-trajectory-from-disturbances}. Since it is also
independent of $\mathcal D$, \cref{thm:cp-confidence-regions} implies
$\prob{s(E^{(0)}) \leq \hat q} \geq 1 - \vartheta$.
Chance constraint satisfaction for the state holds since (where $\forall t$
refers to $t \in \{ 1, \dots, \bar N \}$ for brevity)
\begin{align*}
    \prob{x^{(0)}(t) \in \mathcal X, \forall t}
        &= \prob{z(t) + e^{(0)}(t) \in \mathcal X, \forall t}\\
        &\overset{(\star)}\geq \prob{z(t) + e^{(0)}(t) \in
         \mathcal Z_t \oplus \mathcal E_t, \forall t}\\
        &= \prob{e^{(0)}(t) \in \mathcal E_t, \forall t}\\
        &\overset{\cref{eq:probability-in-tube}}{\geq}
            \prob{E^{(0)} \in \mathcal E}
        \geq 1 - \vartheta.
\end{align*}
We used in $(\star)$ that $(\mathcal X \ominus \mathcal E_t)\oplus \mathcal E_t
\subseteq \mathcal X$. An analogous argument holds for the input constraints.
\end{proof}
\fi

\Cref{thm:constraint-satisfaction} guarantees that conformal bounds for the
error system may be computed offline by propagating the disturbances through the
error system dynamics, choosing a score function, e.g.,
\cref{eq:mahalanobis-score-computation}, and finding the conformal quantile.
Through suitable tightening, joint-in-time chance constraint satisfaction is
established for any trajectory that is feasible for the nominal system. To
further arrive at a recursively feasible formulation, we make use of terminal
ingredients. We define the sets $\mathcal Z_\infty$ and $\mathcal V_\infty$ as
\begin{align}
    \label{eq:z-infty-v-infty}
    \mathcal Z_\infty &\subseteq \bigcap_{t = 1}^{\bar N} \mathcal Z_t
        & \mathcal V_\infty \subseteq \bigcap_{t = 1}^{\bar N} \mathcal V_t.
\end{align}

\begin{assumption}
    \label{ass:terminal-set}
    There exists a non-empty terminal set
    $\mathcal Z_\text{f} \subseteq \mathcal Z_\infty$, with an associated local
    control law $\pi_\text{f}\colon \mathcal Z_\text{f} \to \mathcal V_\infty$
    that renders the terminal set positive invariant such that
    $\forall z \in \mathcal Z_f$, $
        A z + B \pi_\text{f}(z) \in \mathcal Z_\text{f}$.
\end{assumption}

To introduce feedback on the true state, the cost is a function of $x(t)$,
following the uncertain dynamics \cref{eq:dynamics}. The expectation in the cost
\cref{opt:sftoc-cost} can be approximated numerically by using additional
samples of the disturbance trajectory $\{W^{(j)}\}_{j=1}^{S}$.
This results in the optimization problem:
\begin{mini!}|s|
    {V}
    {
        \mathbb E_{\{W^{(j)}\}} \left[
            V_{\mathrm{f}}(x^j_N) + \sum_{i = 0}^{N - 1} \ell(x^j_i, u^j_i)
        \right]
        \label{opt:cpmpc-cost}
    }
    {\label{opt:cpmpc}} 
    {} 
    \addConstraint{\forall i}{= 0, \dots, N - 1}
    \addConstraint{\forall j}{= 1, \dots, S}
    \addConstraint{z_{i + 1}}{= A z_i + B v_i
        \label{opt:cpmpc-dyn-z}
    }
    \addConstraint{e^j_{i + 1}}{= (A + BK) e^j_i + w^{(j)}_{t + i}
        \label{opt:cpmpc-samples}
    }
    \addConstraint{x^j_{i + 1}}{= z_{i + 1} + e^j_{i + 1}
        \label{opt:cpmpc-dyn-x}
    }
    \addConstraint{u^j_i}      {= K e^j_i + v_i
        \label{opt:cpmpc-dyn-u}
    }
    \addConstraint{z_i}        {\in \mathcal Z_{t + i},\quad
                   v_i          \in \mathcal V_{t + i}   
        \label{opt:cpmpc-const-nom}
    }
    \addConstraint{z_N}        {\in \mathcal Z_\text{f}         
        \label{opt:cpmpc-term-z}
    }
    \addConstraint{x_0}        {= x(t), z_0 = z(t), e^j_0 = x_0 - z_0
        \label{opt:cpmpc-init}
    }
\end{mini!}

\begin{theorem}
    \label{thm:if-cpmpc}
    Let $\mathcal E, \mathcal Z_t, \mathcal V_t$ be as in
    \cref{thm:constraint-satisfaction}  for $t = 1, \dots, \bar N$. Further
    assume that \cref{ass:terminal-set} holds. If problem \cref{opt:cpmpc}
    together with the update law $z(t) = z_1^*(t - 1)$ (and $z(0) = x(0)$) is
    feasible at $t = 0$ for $x(0)$, then it is recursively feasible.
    Furthermore, the closed-loop trajectories $x^{(0)}(t)$, $u^{(0)}(t)$ under
    the control law \cref{eq:if-cpmpc-feedback-law} satisfy the joint-in-time
    chance constraints
    \crefrange{eq:chance-constraints-x}{eq:chance-constraints-u}.
\end{theorem}
\ifarXiv
\begin{proof}
    To show recursive feasibility, we follow a standard procedure in the
    literature and use a shifted-sequence argument:
    Let $Z^*(t) = (z_0^*(t), \dots, z_N^*(t))$,
        $V^*(t) = (v_0^*(t), \dots, v_{N - 1}^*(t))$ be the optimizer to
    \cref{opt:cpmpc}
    at time $t$. We prove that the shifted and extended (candidate) solution
    \begin{align*}
        Z(t + 1)    &= \left(z_1^*(t), \dots, z_N^*(t), A z_N^*(t) +
                        B \pi_\text{f}(z_N^*(t))\right),\\
        V(t + 1)    &= \left(v_1^*(t), \dots, v_{N - 1}^*(t),
                        \pi_\text{f}(z_N^*(t))\right)
    \end{align*}
    satisfies \crefrange{opt:cpmpc-const-nom}{opt:cpmpc-term-z}, and is thus
    feasible for \cref{opt:cpmpc} at time $t + 1$: Since the
    constraints are also shifted, $z_0(t + 1), \dots, z_{N - 1}(t + 1)$
    and $v_0(t + 1), \dots, v_{N - 2}(t + 1)$ satisfy
    \cref{opt:cpmpc-const-nom}, with $z_N^*(t) \in \mathcal Z_\text{f} \subseteq
    \mathcal Z_{(t + 1) + (N - 1)}$ by \cref{eq:z-infty-v-infty}.
    By \Cref{ass:terminal-set}, we have that $v_{N - 1}(t + 1) =
    \pi_\text{f}(z_N^*(t)) \in \mathcal V_\infty$ satisfies
    \cref{opt:cpmpc-const-nom} and ensures that
    $z_{N}(t + 1) = A z_N^*(t) + B \pi_\text{f}(z_N^*(t)) \in 
    \mathcal Z_\mathrm{f}$ satisfies \cref{opt:cpmpc-term-z}.
    The shifted sequence is dynamically feasible, fulfilling
    \cref{opt:cpmpc-dyn-z}, and the update law \cref{eq:if-update-law} matches
    \cref{opt:cpmpc-init}. Using $W^{(j)}$,
    \crefrange{opt:cpmpc-samples}{opt:cpmpc-dyn-x} may be constructed, which are
    not subject to constraint bounds.
    Therefore, $Z(t + 1), V(t + 1)$
    are a feasible solution to the optimization problem at time $t + 1$.
    
    Since the problem \cref{opt:cpmpc} is recursively feasible, it holds that
    $z(t) \in \mathcal Z_t$ and $v_0^*(t) \in \mathcal V_t$ for all
    $t = 1, \dots, \bar N$.
    Satisfaction of \crefrange{eq:chance-constraints-x}{eq:chance-constraints-u}
    then follows from \cref{thm:constraint-satisfaction}.
\end{proof}
\fi

Problem \cref{opt:cpmpc} is convex for quadratic costs and ensures closed-loop
chance constraint satisfaction via offline tightening, i.e., at no additional
online computations. Increasing the number of samples in \cref{opt:cpmpc-cost}
improves the approximation but enlarges the problem size, without affecting the
feasible domain.
\begin{remark}
    To deal with the end of the horizon, different approaches may be taken:
    Compute the joint-in-time confidence regions using $\bar N + N$ samples,
    run the optimization problem for $\bar N - N$ time steps, or, use shrinking
    horizon approaches when approaching $\bar N$.
\end{remark}

\subsection{Output Feedback SMPC}
\label{sec:output-feedback-smpc}

If only output measurements are available, the added uncertainty from state
estimation must be considered. The full system dynamics are
\begin{equation}
    \label{eq:output-dynamics}
    \begin{split}
        x(t + 1) &= A x(t) + B u(t) + w(t),\\
        y(t) &= C x (t) + D u(t) + \eta(t),
    \end{split}
\end{equation}
where, similarly, $\eta(t) \sim \mathcal Q^{\eta(t)}$ is drawn from an unknown
distribution. In the following, the approach to solving \cref{opt:sftoc} is
extended to the more general setting of \cref{eq:output-dynamics}, assuming the
availability of i.i.d. samples $\{ \mathcal H^{(k)} \}_{k = 1}^M$ of the
measurement noise trajectories, where
$\mathcal H^{(k)} = (\eta^{(k)}(0), \dots, \eta^{(k)}(\bar N - 1))$.
\begin{remark}
Assuming access to process and measurement noise samples requires separate
disturbance estimation procedures, as in practice only the noisy output is
measurable. However, it allows establishing guarantees for the
observer-controller system, while its relaxation is left for future work.
\end{remark}
A state observer augments the control system, producing the
state estimate $\hat x(t)$. The observer is given by
\begin{equation*}
    \begin{split}
        \hat x(t + 1) &= A \hat x(t) + B u(t) + L(y(t) - \hat y(t)),\\
        \hat y(t) &= C \hat x(t) + D u(t),
    \end{split}
\end{equation*}
where $L$ is such that $A_L := A - LC$ is Schur stable. A well-established
decoupling considers additionally the estimation error
$\hat e(t) = x(t) - \hat x(t)$
and the error between state estimate and a nominal system
$\bar e(t) = \hat x(t) - z(t)$, such that \cite{mayne2006_robust}
\begin{align}
    \notag
    x(t) &= \hat x(t) + \hat e(t) = z(t) + \bar e(t) + \hat e(t),\\
    \label{eq:estimation-error}
    \hat e(t + 1) &= A_L \hat e(t) + w(t) - L \eta(t).\\
    \label{eq:nominal-error-dynamics}
    \bar e(t + 1) &= A_K \bar e(t) + L(C \hat e(t) + \eta(t)).
\end{align}
The two error systems are thus independent of the nominal input and state
trajectory if the pre-stabilizing control law $u(t) = K \bar e(t) + v_0^*(t)$
is employed \cite{mayne2006_robust}. As they only depend on the noise
trajectories $(W, \mathcal H)$, given a known initial condition
$x(0) = \hat x(0) = \bar x(0) = x_0$,
the approach from \cref{sec:indirect-feedback-smpc} can be adapted similarly:
The sample trajectories $\{ (W^{(k)}, \mathcal H^{(k)}) \}_{k = 1}^M$,
propagated through \crefrange{eq:estimation-error}{eq:nominal-error-dynamics},
give rise to trajectories $\{(\hat E^{(k)}, \bar E^{(k)})\}_{k = 1}^M$
which are i.i.d. with the realizations of $(\hat E^{(0)}, \bar E^{(0)})$ in
closed loop. Defining the score function, for all $i = 1, \dots, M$,
\begin{equation}
    \label{eq:mahalanobis-score-combined}
    s(\hat E^{(i)}, \bar E^{(i)}) = \max_{t = 1, \dots, \bar N}
        \norm{\hat e(t) + \bar e(t)}_{\hat \Sigma^{-1}_t},
\end{equation}
where the sample covariance $\hat \Sigma_t$ of the sum
$\hat e(t) + \bar e(t)$ is determined on hold-out data as in
\crefrange{eq:empirical-mean}{eq:empirical-covariance}. This yields a confidence
region $\mathcal E$ such that
\begin{equation*}
    \prob{(\hat e^{(0)}(t) + \bar e^{(0)}(t)) \in \mathcal E_t,
        \forall t = 1, \dots, \bar
N} \geq 1 - \vartheta.
\end{equation*}
We summarize our data-driven approach to the output-feedback setting in the
following optimization problem and then state the resulting guarantees:
\begin{mini!}|s|
    {V}
    {
        \mathbb E_{\{(W^{(j)},\mathcal H^{(j)})\}} \left[
            V_{\mathrm{f}}(x^j_N) + \sum_{i = 0}^{N - 1} \ell(x^j_i, u^j_i)
        \right]
        \label{opt:lqgmpc-cost}
    }
    {\label{opt:lqgmpc}} %
    {}          
    \addConstraint{\forall i}         {= 0, \dots, N - 1}
    \addConstraint{\forall j}         {= 1, \dots, S}
    \addConstraint{{     z}_{i + 1}}  {= A z_i + B v_i}
    \addConstraint{{\hat e^j}_{i + 1}}{= (A - LC)\hat e^j_i + w^{(j)}_{t + i}
                                                        - L \eta^{(j)}_{t + i}
        \label{opt:lqg-est-sample}
    }
    \addConstraint{{\bar e^j}_{i + 1}}{= (A + BK) \bar e^j_i + LC \hat e^j_i
                                                        + L \eta^{(j)}_{t + i}
        \label{opt:lqg-nom-sample}
    }
    \addConstraint{{     x}^j_{i + 1}}{= z_{i + 1} + \hat e^j_{i + 1} +
                                                     \bar e^j_{i + 1}}
    \addConstraint{{     z}_i}        {\in \mathcal Z_{t + i},
                   \quad v_i           \in \mathcal V_{t + i}   
        \label{opt:lqg-const-nom}
    }
    \addConstraint{z_N}               {\in \mathcal Z_\text{f}         
        \label{opt:lqg-term-z}
    }
    \addConstraint{x_0}               {= \hat x(t), z_0 = z(t),
                                         \bar e^j_0 = x_0 - z_0}
\end{mini!}
\begin{theorem}
  \label{thm:lqg-cpmpc}
    Let $\mathcal E$ be a confidence region of level $1 - \vartheta$ for the
    \emph{joint} process $\hat E(t) + \bar E(t)$ and $\mathcal Z_t,
    \mathcal V_t$ be defined analogously to
    \crefrange{eq:tightened-state-constraints}{eq:tightened-input-constraints}
    for $t = 1, \dots, \bar N$.
    Further assume that \cref{ass:terminal-set} holds. If problem
    \cref{opt:lqgmpc} together with the update law $z(t) = z_1^*(t - 1)$
    (and $z(0) = x(0)$), is feasible at $t = 0$ for $x(0)$, then it is
    recursively feasible. Furthermore, under the control law
    $u(t) = K(\hat x(t) - z(t)) + v_0^*(t)$,
    the closed-loop trajectories $x^{(0)}(t)$, $u^{(0)}(t)$ satisfy the
    joint-in-time chance constraints
    \crefrange{eq:chance-constraints-x}{eq:chance-constraints-u}.
\end{theorem}
\ifarXiv
\begin{proof}
    We may use the analogous shifted-sequence argument as in
    \cref{thm:if-cpmpc}. The shifted and extended sequence
    $(z_1^*, \dots, z_N^*, A z_n^* + B \pi_\text{f}(z_N^*(t))$ and
    $(v_1^*, \dots, v_N^*, \pi_\text{f}(z_N^*(t))$ is a feasible solution to
    \cref{opt:lqgmpc} and satisfies constraints at time $t + 1$, as the
    constraints are also shifted and because \cref{ass:terminal-set} holds
    for the two final entries.
    
    Since the problem \cref{opt:cpmpc} is recursively feasible, it holds that
    $z(t) \in \mathcal Z_t$ and $v_0^*(t) \in \mathcal V_t$ for all
    $t = 1, \dots, \bar N$. Satisfaction of
    \crefrange{eq:chance-constraints-x}{eq:chance-constraints-u} then follows
    from an analogous argument to \cref{thm:constraint-satisfaction}:
    The realizations of the random variables $\hat E^{(0)}, \bar E^{(0)}$
    are i.i.d. with the trajectories in the calibration dataset used to
    find the conformal confidence region $\mathcal E$, so it holds that
    $\prob{\hat E^{(0)} + \bar E^{(0)} \in \mathcal E} \geq 1 - \vartheta$.
    It suffices to note that
    (where $\forall t$ denotes $t \in \{1, \dots, \bar N \}$)
    \begin{equation*}
        \begin{split}
            \prob{x(t) \in \mathcal X, \forall t} &=
                \prob{z(t) + \bar e(t) + \hat e(t) \in \mathcal X,~\forall t}\\
                &\geq \prob{\bar e(t) + \hat e(t) \in \mathcal E_t,~\forall t}
                \geq 1 - \vartheta,
        \end{split}
    \end{equation*}
    since $(\mathcal X \ominus \mathcal E_t) \oplus
            \mathcal E_t \subseteq \mathcal X$.
    The analogous argument holds for the input constraints.
\end{proof}
\fi

In \cref{opt:lqgmpc-cost}, we again leverage additional samples
$\{(W^{(j)}, \mathcal H^{(j)})\}_{j=1}^{S}$ to approximate the expectation
\cref{opt:sftoc-cost}.

\section{SIMULATION EXAMPLE}
\label{sec:simulations}
\begin{figure*}[h]
    \centering
        \begin{subfigure}{0.22\textwidth}
            \includegraphics[width=\textwidth]
                {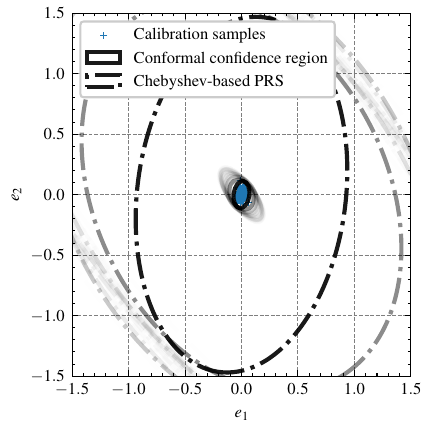}
            \caption{}
            \label{fig:confidence-regions}
            \vspace{-2mm}
        \end{subfigure}
        \hfill
        \begin{subfigure}{0.22\textwidth}
            \includegraphics[width=\textwidth]
                {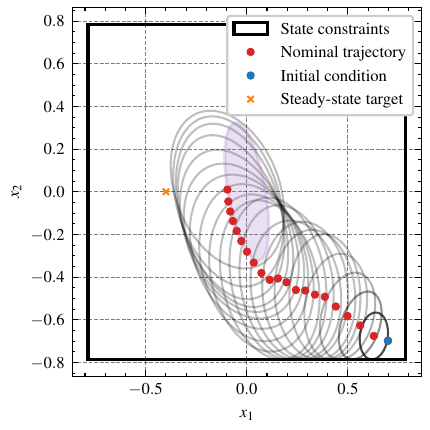}
            \caption{}
            \label{fig:nominal-trajectory}
            \vspace{-2mm}
        \end{subfigure}
        \hfill
        \begin{subfigure}{0.22\textwidth}
            \includegraphics[width=\textwidth]
                {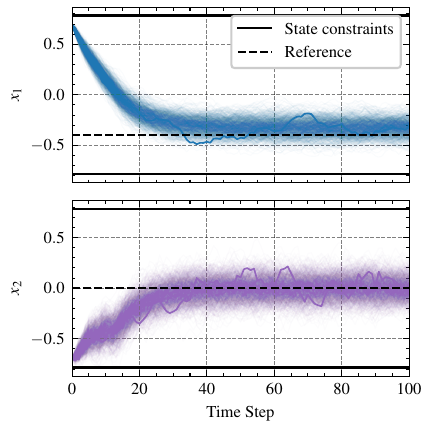}
            \caption{}
            \label{fig:cl-trajectory}
            \vspace{-2mm}
        \end{subfigure}
        \hfill
        \begin{subfigure}{0.22\textwidth}
            \includegraphics[width=\textwidth]
                {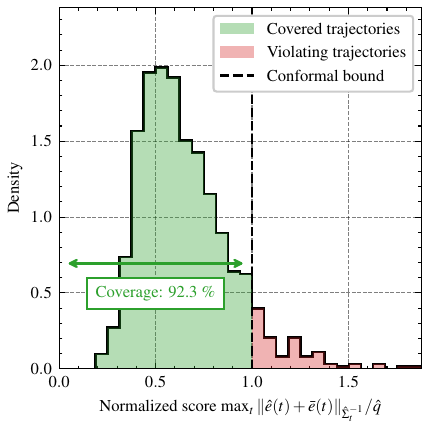}
            \caption{}
            \label{fig:conformal-scores}
            \vspace{-2mm}
        \end{subfigure}
    \caption{a) Confidence regions for the state feedback case with
             $\vartheta = 0.1$ for $t = 1, \dots, \bar N$. Higher opacity
             indicates later time steps.
             b) Nominal trajectory and projected confidence regions at $x_0$
                for $t = 0$ for the output feedback case.
             c) Closed-loop trajectories, one realization
                emphasized.
             d) Empirical distribution of newly drawn error trajectory scores,
                normalized by $\hat q$.}
    \vspace{-6mm}
\end{figure*}
We demonstrate the proposed approach in a numerical example on an open-loop
unstable linearized inverted pendulum and compare the method to solving the
stochastic optimal control problem \cref{opt:sftoc}, applying the tube policy
proposed in~\cite{vlahakis2024_conformal}. Further, we compare the size of the
prediction regions obtained from our approach to PRS based on mean-variance
information obtained from data. The system dynamics are given by
\cref{eq:output-dynamics}, subject to tracking costs
$\ell(x, u) = \norm{x - x_s}_Q^2 + \norm{u - u_s}_R^2$,
with $(x_s, u_s)$ a steady state of~\cref{eq:output-dynamics}, corresponding
terminal cost $V_\text{f}(x) = \| x - x_s \|_P^2$, and constraints
$\mathcal X, \mathcal U$. We set $K$ as the infinite-horizon LQR controller for
$Q, R$ and obtain $P$ from the discrete-time Riccati equation.
The goal is for a maximum violation probability of $\vartheta = 0.1$ to hold
over a horizon $\bar N = 100$, where the MPC horizon is $N = 20$.
The optimization problems \cref{opt:cpmpc} and \cref{opt:lqgmpc}  are
implemented using the \emph{ampyc} framework \cite{ampyc}. All numerical values
are in
\ifarXiv
    the \hyperref[sec:appendix]{Appendix}.
\else
    \cite[Appendix]{vogel2025conformalmpc}.
\fi

\paragraph{Confidence region comparison}

In the state feedback case, we construct confidence regions of $1 - \vartheta$
over $\bar N + N$ time steps based on $M = 750$ disturbance trajectories from a
mixture of two Gaussian distributions, partitioned into
$\abs{\mathcal D_\text{fit}} = 250$ and $\abs{\mathcal D_\text{cal}} = 500$. The
sample covariance $\hat\Sigma_t$ is computed by \cref{eq:empirical-covariance}
for $t = 1, \dots, \bar N + N$, assuming w.l.o.g. that $\hat \mu_t = 0$.
\Cref{thm:cp-confidence-regions} is used with
\cref{eq:mahalanobis-score-computation} to obtain $\hat q$, such that
$\mathcal E^\text{cp}_t=\{\,e \mid e^\top \hat \Sigma^{-1}_t e \leq \hat q\,\}$.
We compare against mean-variance PRS \cite{hewing2020_recursively} and obtain
$\mathcal E^\text{mv}_t = \{\,e \mid
    e^\top \hat \Sigma^{-1}_t e \leq \tilde p\,\},$
where $\tilde p = n_x / (\vartheta/(\bar N + N))$ by the Chebyshev inequality
and Boole's inequality.
\Cref{fig:confidence-regions} depicts the two confidence regions for all time
steps, with lighter shades indicating later time steps. The Chebyshev bound
is---as a worst-case bound---more conservative than the conformal bound,
in particular due to the required union bound over large $\bar N + N$.

\paragraph{Closed-Loop Performance}

Next, we use the obtained conformal confidence region and solve the optimal
control problem with state feedback, comparing to \cite{vlahakis2024_conformal}.
We initialize the system at ${x_0 = [0.70, -0.66]^\top}$, and, for identical
confidence regions, solve the problem once at $t = 0$ for the horizon
$\bar N$, applying the control law $u(t) = K(x(t) - z_t^*(0)) + v_t^*(0)$.
For our proposed approach, we solve the problem in receding horizon of $N$ and
apply $u(t) = K(x(t) - z_0^*(t)) + v_0^*(t)$. Compared with
\cite{vlahakis2024_conformal}, we achieve an average cost reduction of 1.9~\%
over $n_\text{test} = 10^3$ test trajectories with a shorter horizon.

\paragraph{Output Feedback SMPC}

Lastly, output feedback under measurement noise drawn from a Laplace
distribution is considered. An additional $M$ measurement noise trajectories are
used to find i.i.d. samples of the errors $\bar e(t), \hat e(t)$ under the fixed
observer gain $L$. The system and the observer share the same initial condition
$\hat x(0) = x(0) = x_0$, and evolve in closed loop for $n_\text{test} = 10^3$
realizations of the noise trajectories over time. In
\Cref{fig:nominal-trajectory}, the nominal trajectory computed at time $t = 0$
and the confidence regions are shown. Further, multiple realizations of the
closed-loop trajectory $x(t)$ are depicted in \Cref{fig:cl-trajectory}.
Closed-loop satisfaction of
\crefrange{eq:chance-constraints-x}{eq:chance-constraints-u} is 99\,\%,
achieving the desired level at some conservatism due to the system tracking a
point in the interior and not at the boundary. \Cref{fig:conformal-scores}
depicts the empirical distribution of the conformal scores achieved by the
realized trajectories, indicating that 92.3\,\% of newly drawn
error trajectories are covered by the probabilistic bound, emphasizing the
non-conservative nature of conformal prediction.

\section{CONCLUSION}

In this work, we presented a stochastic MPC formulation for linear systems
subject to joint-in-time chance constraints and disturbances from an unknown
distribution. The proposed approach uses available samples to construct
conformal confidence bounds for error trajectories, enabling the handling of
joint chance constraints in a relaxation of the underlying stochastic program.
We proved recursive feasibility of the relaxed problem for both the state and
output feedback cases, guaranteeing chance constraint satisfaction in
closed loop. Finally, the guarantees were verified in a simulation example,
showcasing that the error confidence regions are non-conservative
while satisfying the desired probability level.


\ifarXiv
\section*{APPENDIX}
\label{sec:appendix}
In the following, we provide some more detail on the simulation example from
\cref{sec:simulations}. The system follows the dynamics
\cref{eq:output-dynamics}, where
\begin{align*}
    A &= \begin{bmatrix}
        1.00 & 0.10\\
        0.75 & 0.95
    \end{bmatrix}, & B &= \begin{bmatrix}
        0.00\\
        0.10
    \end{bmatrix}, &
    C &= \begin{bmatrix}
        1 & 0
    \end{bmatrix}, & D &= 0.
\end{align*}
The state and input constraints are given by
$\mathcal X = [-\frac{\pi}{4}, \frac{\pi}{4}]^2$, and $\mathcal U = [-8, 8]$,
respectively.

We choose $Q = 10^3 \cdot \mathbb I_2$, $R = 10$ and solve the DARE for
$(A, B, Q, R)$ for $P$ and obtain $K$ as
\begin{equation*}
    K = -(R + B^\top P B)^{-1} B^\top P A = \begin{bmatrix}
        -14.87 & -7.24
    \end{bmatrix}.
\end{equation*}
The tracked steady-state is $x_s = [-0.5, 0.0]^\top$ and
$u_s = B^\dagger (x_s - A x_s) = 3.75$.

\paragraph*{Noise distributions}
The process and measurement noise distributions are
\begin{align*}
    w(t) &\sim \frac{1}{2}\mathcal N(\mu_w, \Sigma_w) +
        \frac{1}{2}\mathcal N(-\mu_w, \Sigma_w)\\
    \eta(t) &\sim \operatorname{Laplacian}(\mu_\eta, b_\eta),\\
    \mu_w &= [0.01, 0.001]^\top,
        \quad \Sigma_w = 10^{-3} \cdot \operatorname{diag}(0.1, 1)\\
    \mu_\eta &= 0, \quad b_\eta = 5 \cdot 10^{-4}.
\end{align*}
The implementation is available at
\url{https://gitlab.ethz.ch/ics/conformal-prediction-smpc}.

\fi


\bibliographystyle{IEEEtran}
\bibliography{bibliography}

\end{document}